\documentclass{article}

\usepackage{arxiv}
\usepackage{placeins}  
\usepackage[utf8]{inputenc} 
\usepackage[T1]{fontenc}    
\usepackage{hyperref}       
\usepackage{url}            
\usepackage{booktabs}       
\usepackage{amsfonts}       
\usepackage{nicefrac}       
\usepackage{microtype}      
\usepackage{lipsum}
\usepackage{graphicx}
\usepackage{diagbox}   
\usepackage{multirow}  
\usepackage{booktabs}  
\usepackage{makecell}
\usepackage{array} 
\usepackage[utf8]{inputenc}
\usepackage{algorithm}
\usepackage{algpseudocode}
\usepackage{hyperref}
\usepackage{authblk}
\usepackage{enumitem}
\usepackage{mathtools}
\usepackage{amssymb}
\usepackage{amsmath}
\usepackage{setspace}
\usepackage{natbib}
\usepackage{amsthm}
\usepackage{graphicx}
\usepackage{orcidlink}
\newtheorem{theorem}{Theorem}
\newtheorem{definition}[theorem]{Definition}

\newtheorem{lemma}[theorem]{Lemma}
\newtheorem{remark}[theorem]{Remark}
\graphicspath{ {./images/} }

\title{Debiased Estimators in High-Dimensional Regression: A Review and Replication of Javanmard and Montanari (2014)}

\author{
 Benjamin Smith\orcidlink{0009-0007-2206-0177} \\
  Department of Statistical Sciences\\
  University of Toronto\\
  Toronto, Canada \\
  \texttt{benyamin.smith@mail.utoronto.ca} \\
}

\begin{document}
\maketitle
\begin{abstract}
High-dimensional statistical settings ($p \gg n$) pose fundamental challenges for classical inference, largely due to bias introduced by regularized estimators such as the LASSO. To address this, \citet{javanmard14a} propose a debiased estimator that enables valid hypothesis testing and confidence interval construction. This report examines their debiased LASSO framework, which yields asymptotically normal estimators in high-dimensional settings. The key theoretical results underlying this approach are presented. Specifically, the construction of an optimized debiased estimator that restores asymptotic normality, which enables the computation of valid confidence intervals and $p$-values. To evaluate the claims of Javanmard and Montanari, a subset of the original simulation study and the real-data analysis is presented. The original empirical analysis is extended to the desparsified LASSO, which is referenced but not implemented in the original study. The results demonstrate that while the debiased LASSO achieves reliable coverage and controls Type I error, the LASSO projection estimator can offer improved power in idealized low-signal settings without compromising error rates. The results reveal a trade-off: the LASSO projection estimator performs well in low-signal settings, while Javanmard and Montanari’s method is more robust to complex correlations, improving precision and signal detection in real data.
\end{abstract}


\newpage
\section{Introduction}

Modern statistical problems are widely recognized as having become increasingly high-dimensional \citep{javanmard14a}. Examples such as genomics \citep{Peng_Zhu_Bergamaschi_Han_Noh_Pollack_Wang_2010}, neuroimaging \citep{Yamashita2008}, asset pricing \citep{Kozak_Nagel_Santosh_2019}, signal processing \citep{lustig2007sparse} and collaborative filtering\footnote{Collaborative filtering is a core technique used in recommender systems that suggests items by finding patterns in user behaviour. It works under the assumption that individuals who agree in the past will agree in the future \citep{Al_Ghuribi_Noah_2021}}\citep{koren2009} have led to settings in which the number of covariates exceeds the number of observations ($p > n$). LASSO regression \citep{Tibshirani_1996} has become a widely applicable and successful technique for addressing such problems. Estimators are derived by optimizing a regularized likelihood function, typically incorporating an $L_1$ penalty to induce sparsity. Because they are defined as the solution to a constrained optimization problem, LASSO estimators are inherently non-linear and do not admit an explicit, closed-form representation.

The non-linear and non-explicit nature of the estimators produced by LASSO estimates results in being unable to characterize the distribution the estimators it produces \citep{Fu_Knight_2000}. This creates a challenge inference as one is unable to preform accepted procedures for characterizing uncertainty via confidence intervals and p-values \citep{wasserman2004all, lehmann2005testing}. As of 2014\footnote{This was the publication date of \citet{javanmard14a} paper.} approaches to quantify statistical significance in high-dimensional parameter estimation settings has been less understood. \citet{Buhlman2013} and \citet{zhang2014confidence} introduced high-dimensional testing frameworks that utilize compatibility conditions to provide deterministic guarantees for significance $\alpha$ and power $1-\beta$ under restricted eigenvalue or compatibility conditions \citep{buhlmann2011statistics}. These results were later strengthened by \citet{javanmard_nips2013}, who established more refined lower bounds for the signal detection threshold.  \citet{Belloni_2013} and  \citet{Belloni_Chernozhukov_Hansen_2014} consider inference for regression models with high-dimensional data. \citet{Meinshausen_2010} and \cite{Minnier_Tian_Cai_2011} studied resampling methods for hypotheses testing for these methods. 

This report reviews and discusses the theoretical framework introduced by \citet{javanmard14a} in their Journal of Machine Learning Research paper titled ``Confidence Intervals and Hypothesis Testing for High-Dimensional Regression''. Specifically, it examines the key lemmas and theorems underlying their approach to debiased inference for LASSO regression and presents a partial replication\footnote{Due to resource constraints, only a subset of the original simulation study is replicated.} of the original simulation study and an analysis of the data used in the original data analysis. The analysis compares their proposed method with the multisample splitting method \citep{Meinshausen_Buhlman_2009, meinshausen2010stability} and the ridge-type projection estimator \citep{Buhlman2013}, as well as with alternative approaches implemented in the \texttt{hdi} R package \citep{hdi_pkg2015}. Furthermore, the comparison is extended to include the desparsified LASSO  estimator (also known as the LASSO projection estimator, used interchangably) \citep{vandegeer2014asymptotically, zhang2014confidence, buhlmann2015misspecified}, a closely related approach to debiased inference that was not formally evaluated (but referenced) in the original paper.

In Section \ref{sec:methodology}, we outline the model setting and the de-biasing algorithm proposed by \citet{javanmard14a}, followed by an overview of technical prerequisites defined in Section \ref{sec:2.1}. Section \ref{sec:2.2} presents the key theoretical results required for formal inference with the debiased estimator. Section \ref{sec:2.3} provides an overview of the primary applications derived from the theorems, including the construction of confidence intervals, hypothesis tests, and extensions to simultaneous inference and settings involving non-Gaussian noise. Section \ref{sec:empirical_analysis} assesses the empirical performance of \citet{javanmard14a}'s debiased LASSO estimator through a partial replication of their simulation studies (Section \ref{sec:3.1}) and a analysis of the riboflavin dataset \citep{buhlmann2014high}. By examining coverage probabilities, interval lengths, and Type I error rates, we compare the proposed method against competing approaches noted, but extend the analyses to include the LASSO projection estimator in order to highlight practical trade-offs. Section \ref{sec:discussion} concludes by combining the theoretical and empirical findings, discussing the practical trade-offs of debiased inference and its implications for high-dimensional statistical analysis in comparison with competing methods such as the LASSO projection estimator.

For reproducibility, the \texttt{R} code used to perform the simulations and data analysis is available at:\\\url{https://github.com/benyamindsmith/javanmard_montanari_2014_replication}

\section{Methodology}
\label{sec:methodology}
We work in a linear regression model setting where $n$ i.i.d.  pairs $(Y_1, X_1), (Y_2,X_2), \dots, (Y_n, X_n)$ are given with vectors $X_i \in \mathbb{R}^p$ and response variables defined as
\begin{equation}
    Y_i= \langle\theta_0, X_i\rangle + W_i, \quad W_i \sim N(0, \sigma^2)
    \label{eq:mod_1dim}
\end{equation}
Where $\theta_0\in \mathbb{R}^p$ and $\langle\cdot, \cdot\rangle$ is the standard scalar product in $\mathbb{R}^p$. In matrix form, denote $Y=(Y_1, \dots, Y_n)^\top$ and $\mathbf{X} \in \mathbb{R}^{n\times p}$. This gives
\begin{equation}
    Y = \mathbf{X}\theta_0 + W, \quad W \sim N(0, \sigma^2 I_{n\times n})
    \label{eq:mod_mat_form}
\end{equation}
The goal is to estimate the unknown (but fixed) parameter vector $\theta_0\in \mathbb{R}^p$. In a classical regression setting where $n \gg p$, the estimation method of choice is usually ordinary least squares, i.e. $\hat{\theta}^{\text{OLS}}= (\mathbf{X}^\top\mathbf{X})^{-1}\mathbf{X}^\top Y$. By definition, $\hat{\theta}^{\text{OLS}}$ is Gaussian with mean $\theta_0$ and and covariance $\sigma^2(\mathbf{X}^\top \mathbf{X})^{-1}$. This allows for typical inference techniques such as confidence intervals, hypothesis tests and construction of P-values. 

In a high dimensional setting where the number of parameters exceeds the number of observations ($p > n$), the matrix $(X^\top X)$ is rank deficient and one is required to resort to biased estimators \citep{buhlmann2011statistics}. The LASSO \citep{chen_donho_1995,Tibshirani_1996} has become a widely used method that produces sparse estimates by incorporating an $\ell_1$ penalty into the objective function:
\begin{equation}
    \hat{\theta}^n(Y, \mathbf{X};\lambda) = \arg \underset{\theta \in \mathbb{R}^p}{\min} \left\{\frac{1}{2n}\|Y - \mathbf{X}\theta\|^2_2 + \lambda \|\theta\|_1\right\}
\label{eq:lasso_eq}
\end{equation}

The $\ell_1$ penalty used in LASSO shrinks coefficients toward zero, thus introducing bias into the estimators. While this promotes sparsity, it also complicates statistical inference for the selected model as the resulting sampling distribution is difficult to characterize. Algorithm 1 shows the algorithm proposed by \citet{javanmard14a} to construct the debiased estimator from the LASSO, given by Equation (\ref{eq:lasso_debiased}) as $\hat{\theta}^{u}=\hat{\theta}^{\,n}(\lambda)+\frac{1}{n}M\mathbf{X}^{\top}\bigl(Y-\mathbf{X}\hat{\theta}^{\,n}(\lambda)\bigr)$. The basic intuition is that the subgradient of the $\ell_1$ norm at the LASSO solution $\hat{\theta}^n$ is $\mathbf{X}^\top(Y - \mathbf{X}\hat{\theta}^n)/(n\lambda)$. By adding a term which is proportional to the subgradient, (\ref{eq:lasso_debiased}) compensates the bias introduced by the $\ell_1$ penalty in the LASSO.

\begin{algorithm}
\caption{Unbiased estimator for $\theta_0$ in high-dimensional linear regression models}
\begin{algorithmic}[1]
\Statex\textbf{Input:} Measurement vector $y$, design matrix $\mathbf{X}$, parameters $\lambda,\mu$.
\Statex\textbf{Output:} Unbiased estimator $\hat{\theta}^{u}$.
\State Let $\hat{\theta}^{\,n}=\hat{\theta}^{\,n}(y;\mathbf{X};\lambda)$ be the LASSO estimator as per (\ref{eq:lasso_eq}).
\State Set $\widehat{\Sigma}=(\mathbf{X}^{\top}\mathbf{X})/n$.
\For{$i=1,2,\dots,p$}
    \State Let $m_i$ be a solution of the convex program:
    \begin{equation}
    \begin{aligned}
    &\text{minimize} && m^\top \widehat{\Sigma} m \\
    &\text{subject to} && \|\widehat{\Sigma}m-e_i\|_\infty \le \mu,
    \end{aligned}
    \label{eq:convex_optimization}
    \end{equation}
    where $e_i\in\mathbb{R}^p$ is the vector with one at the $i$-th position and zero everywhere else.
\EndFor
\State Set $M=(m_1,\dots,m_p)^\top$. If any of the above problems is not feasible, then set $M=I_{p\times p}$.
\State Define the estimator $\hat{\theta}^{u}$ as follows:
\begin{equation}
    \hat{\theta}^{u}
=
\hat{\theta}^{\,n}(\lambda)
+\frac{1}{n}M\mathbf{X}^{\top}\bigl(Y-\mathbf{X}\hat{\theta}^{\,n}(\lambda)\bigr)
\label{eq:lasso_debiased}
\end{equation}
\end{algorithmic}
\end{algorithm}

It is claimed by \citet{javanmard14a} that the estimator (\ref{eq:lasso_debiased}) guarantees near optimal confidence interval sizes and testing power. Section \ref{sec:2.1} presents preliminary definitions necessary for understanding the core theory presented in \citet{javanmard14a}. Section \ref{sec:2.2} presents the theorems which demonstrate that $\hat{\theta}^{u}$ is Gaussian and is an asymptotically unbiased estimator for $\theta_0$. Section \ref{sec:2.3} shows that the results from theorems in section \ref{sec:2.2} can be applied directly to derive confidence intervals (section \ref{sec:2.3.1}), show near optimality of hypotheses tests generated (section \ref{sec:2.3.2}),along with being generalized to derivation of simultaneous confidence intervals (Section \ref{sec:2.3.3}) and the validity of the estimator $\hat{\theta}^u$ in non-Gaussian noise settings (Section \ref{sec:2.3.4}).


\subsection{Prerequisites}
\label{sec:2.1}
In this section, we will present basic definitions which will be used later in this paper.
\begin{definition}[True Support Set of $\theta_0$]
We denote by $S \equiv \text{supp}(\theta_0)$ the support of $\theta_0 \in \mathbb{R}^p$, defined as: 
\begin{equation}
\label{eq:supp_theta}
    \text{supp}(\theta_0) \equiv \left\{i \in [p]: \theta_{0,i} \neq 0 \right\}
\end{equation}
Where $[p] = \left\{1, \dots, p \right\}$ is the index set of all parameters. 
\end{definition}
\noindent
In words, $\text{supp}(\theta_0)$ is the set of non-zero sparse parameters in the model (\ref{eq:mod_1dim}). We further denote $s_0 \equiv |S|$ to be the number of elements in $\text{supp}(\theta_0)$.

\begin{definition}[Sub-Gaussian Norm]
The sub-Gaussian norm of a random variable $X$, denoted by $\|X\|_{\psi_2}$ is defined as
$$
\|X\|_{\psi_2} = \sup_{q \ge 1} q^{-1/2} (\mathbb{E}|X|^q)^{1/q}
$$
For a random vector $X \in \mathbb{R}^n$, its sub-Gaussian norm is defined as
$$
\|X\|_{\psi_2} = \sup_{x \in \mathbb{S}^{n-1}}\|\langle X, x\rangle\|_{\psi_2}
$$
Where $\mathbb{S}^{n-1}$ denotes the unit sphere in $\mathbb{R}^n$
\end{definition}
\noindent
The next two definitions provide the necessary foundation for the theorems in Section \ref{sec:2.2}, which establish the unbiasedness of $\hat{\theta}^{u}$.

\begin{definition}[Compatibility Constant]
Given a symmetric matrix $\hat{\Sigma} \in \mathbb{R}^{p\times p}$ and a set $S \subseteq [p]$, the corresponding compatibility constant is defined as
\begin{equation}
    \phi^2(\hat{\Sigma}, S) \equiv \underset{\theta\in \mathbb{R}^p}{\min} \left\{\frac{|S| \langle \theta, \hat{\Sigma}\theta \rangle}{\|\theta_S\|^2_1}: \theta \in \mathbb{R}^p, \|\theta_{S^c}\|_1 \le 3 \|\theta_S \|_1\right\}
\end{equation}
We say that $\hat{\Sigma} \in \mathbb{R}^{p\times p}$ satisfies the compatibility condition for $S \subseteq [p]$ with constant $\phi_0$ if $\phi(\hat{\Sigma}, S) \ge \phi_0$. We say it holds for the design matrix $\mathbf{X}$ if it holds for $\hat{\Sigma} = \mathbf{X}^\top\mathbf{X}/n$.
\end{definition}

\begin{definition}[Generalized Coherence Parameter]
Given the pair $\mathbf{X}\in \mathbb{R}^{n\times p}$ and $M \in \mathbb{R}^{p\times p}$. Let $\hat{\Sigma} = \mathbf{X}^\top\mathbf{X}/n$ denote the associated sample covariance. The generalized coherence parameter of $\mathbf{X}, M$ is: 
\begin{equation}
    \mu_*(\mathbf{X}; M) \equiv |M \hat{\Sigma} - I|_\infty = \max_{i\ge1, j\le p}|(M\hat{\Sigma})_{ij} - I_{ij}| = \max \left\{\max_{i}|(M\hat{\Sigma})_{ii} - 1|, \max_{i\neq j} |(M\hat{\Sigma})_{ij}|\right\}
\end{equation}
We denote the minimum (generalized) coherence of $\mathbf{X}$ as $\mu_{\min}(\mathbf{X}) = \min_{M\in \mathbb{R}^{p\times p}}\mu_*(\mathbf{X}; M)$ and any minimizer of $\mu_*(\mathbf{X}; M)$ as $M_{\text{min}}(X)$. 
\end{definition}
\noindent
Note that the minimum coherence can be computed efficiently since $M \mapsto \mu_*(\mathbf{X}; M)$ is a convex function.
\begin{remark}
If we assume that the columns of of $\mathbf{X}$ are normalized to have the $\ell_2$ norm equal to $\sqrt{n}$ (i.e., $\|\mathbf{X} e_i\| = \sqrt{n}$ for all $i \in [p]$), and $M=I$, then $(M\hat{\Sigma} - I)_{i,i}= 0$ and $|M\hat{\Sigma} - I|_\infty = \max_{i\neq j}|(\hat{\Sigma})_{i,j}|$. In other words, $\mu(\mathbf{X}; I)$ is the maximum normalized scalar product between distinct columns of $\mathbf{X}$. i.e.
\begin{equation}
    \mu_*(\mathbf{X}; I) = \frac{1}{n} \max_{i\neq j}| \langle \mathbf{X}e_i, \mathbf{X}e_j\rangle|
    \label{eq:normed_coherence}
\end{equation}
Then the quantity (\ref{eq:normed_coherence}) is known as the coherence parameter of the matrix $\mathbf{X}/\sqrt{n}$.
\end{remark}
\noindent
Assuming for the sake of simplicity that the columns of $\mathbf{X}$ are normalized so that $\|\mathbf{X}e_i\|_2 = \sqrt{n}$, a small value of the coherence parameter $\mu_*(\mathbf{X}; I)$ means that the columns of $\mathbf{X}$ are roughly orthogonal. So $\mu_*(\mathbf{X}; I) =0$ if and only if $\mathbf{X}/\sqrt{n}$ is column orthogonal. However, the generalized coherence parameter $\mu_*(\mathbf{X}; M)$ can be much smaller (i.e. it can remain small, potentially as small as order $O(\sqrt{(\log p)/n})$, even when the columns of $\mathbf{X}$ are significantly correlated). 

Armed with these definitions, we now present the central theorems that establish the asymptotic normality of the debiased estimator $\hat{\theta}^u$.

\subsection{A Debiased Estimator for $\theta_0$}
\label{sec:2.2}
Having established the generalized coherence $\mu_*$, we now characterize the  distributional properties of the debiased estimator. The following theorems  provide the conditions under which $\hat{\theta}^u$ is asymptotically normal, which will allow for classical inferential procedures such as confidence intervals and hypotheses testing. In order to clarify the distributional properties of $\hat{\theta}^u$, \citet{javanmard14a} consider a general debiasing procedure that makes use of an arbitrary debiasing matrix $M \in \mathbb{R}^{p\times p}$ and define the \textit{generalized debiased estimator} $\hat{\theta}^*$ as: 
\begin{equation}
    \hat{\theta}^*(Y, \mathbf{X}; M, \lambda) = \hat{\theta}^n(\lambda) + \frac{1}{n}M\mathbf{X}^\top (Y-\mathbf{X}\hat{\theta}^n(\lambda))
    \label{eq:gen_debiased_theta}
\end{equation}
\begin{theorem}[Error Decomposition around $\hat{\theta^*}$]
\label{thm:thm6}
Let $\mathbf{X} \in \mathbb{R}^{n\times p}$ be any (determinisitic) design matrix and $\hat{\theta}^* = \hat{\theta}^*(Y, \mathbf{X}; M,\lambda)$ be the generalized debiased estimator as per Equation (\ref{eq:gen_debiased_theta}). 
\begin{enumerate}
    \item Setting  $Z= M \mathbf{X}^\top W/\sqrt{n}$ gives
    \begin{equation}
        \sqrt{n}(\hat{\theta}^*- \theta_0) = Z + \Delta, \quad Z \sim N(0, \sigma^2 M \hat{\Sigma}M^\top), \quad \Delta = \sqrt{n}(M\hat{\Sigma} - I)(\theta_0 - \hat{\theta}^n)
    \end{equation}
    \item Further, assume that $\mathbb{X}$ satisfies the compatibility condition for the set $S = \text{supp}(\theta_0), |S| \le s_0$, with constant $\phi_0$ and has a generalized coherence parameter $\mu_* = \mu_*(\mathbf{X}; M)$ and let $K \equiv \max_{i \in [p]}\hat{\Sigma}_{i,i}$. Then letting $\lambda = \sigma \sqrt{(c^2 \log p)/n}$ gives
    \begin{equation}
    \label{eq:thm6eq2}
        \mathbb{P}\left(\|\Delta\|_\infty \ge \frac{4c \mu_* \sigma s_0}{\phi_0^2} \sqrt{\log p}\right) \le 2p^{-c_0}, \quad c_0 = \frac{c^2}{32K}-1
    \end{equation}
    Note, if $M= M_{min}(X)$ minimizes $\mu_*(\mathbf{X},M) = |M\hat{\Sigma} - I|_\infty$, then $\mu_*$ can be replaced by $\mu_\text{min}(X)$ in Equation (\ref{eq:thm6eq2}).
\end{enumerate}
\end{theorem}
\noindent
Theorem \ref{thm:thm6} decomposes the estimation error of the generalized debiased estimator $(\hat{\theta}^*- \theta_0)$ into a zero mean Gaussian term $Z/\sqrt{n}$ and a bias term $\Delta/\sqrt{n}$ whose maximum entry is bounded in probability per Equation \ref{eq:thm6eq2}.
\begin{proof}
\leavevmode\par
\begin{enumerate}
\item By substituting $Y= X\theta_0 + W$ into Equation (\ref{eq:gen_debiased_theta}), we can show that the estimation error $(\hat{\theta}^* - \theta_0)$  can be decomposed into the zero mean Gaussian term $Z/\sqrt{n}$ and the bias term $\Delta/\sqrt{n}$. 
\begin{align*}
    \hat{\theta}^* &= \hat{\theta}^n + \frac{1}{n} M \mathbf{X}^\top(Y-X \hat{\theta}^n)\\
    &=\hat{\theta}^n + \frac{1}{n}M\mathbf{X}^\top(\mathbf{X} \theta_0 + W - \mathbf{X}\hat{\theta}^n)\\
    &= \hat{\theta}^n + \frac{1}{n}M\mathbf{X}^\top\mathbf{X}(\theta_0 - \hat{\theta}^n) + \frac{1}{n}M \mathbf{X}^\top W \\
    &= \hat{\theta}^n + \frac{1}{n}M\mathbf{X}^\top\mathbf{X}(\theta_0 -\hat{\theta}^n) + \frac{1}{n}M\mathbf{X}^\top W - (\theta_0 - \hat{\theta}^n) + (\theta_0 - \hat{\theta}^n) \\
    &=\theta_0 + \left(\frac{1}{n}M\mathbf{X}^\top \mathbf{X} - I\right)(\theta_0 - \hat{\theta}^n) + \frac{1}{n}M \mathbf{X}^\top W\\
    &= \theta_0 + (M\hat{\Sigma}-I)(\theta_0 - \hat{\theta}^n) + \frac{1}{\sqrt{n}}\cdot\frac{1}{\sqrt{n}}M\mathbf{X}^\top W\\
    &=\theta_0 + \frac{1}{\sqrt{n}}\Delta + \frac{1}{\sqrt{n}}Z \\
    &=\theta_0 + \frac{1}{\sqrt{n}}Z + \frac{1}{\sqrt{n}}\Delta\\
    \iff& (\hat{\theta}^* - \theta_0) = + \frac{1}{\sqrt{n}}Z + \frac{1}{\sqrt{n}}\Delta
\end{align*}
\item Noting that  
\begin{align*}
\|\Delta\|_\infty &= \|\sqrt{n}(M\hat{\Sigma} - I)(\theta_0 - \hat{\theta}^n)\|_\infty\\
(\substack{\text{Triangle}\\\text{Inequality}})&\le\sqrt{n}\|M\hat{\Sigma}-I\|_\infty \|\theta_0 - \hat{\theta}^n\|_1 \\
&= \sqrt{n}\mu_*(\mathbf{X}; M) \|\theta_0 - \hat{\theta}^n\|_1\\
&= \sqrt{n}\mu_*(\mathbf{X}; M) \| \hat{\theta}^n-\theta_0 \|_1
\end{align*}
By \citet[Theorem 6.1, Lemma 6.2]{buhlmann2011statistics},for any $\lambda \ge 4\sigma \sqrt{2K \log(p e^{t^2/2})/n}$ 
\begin{equation}
    \mathbb{P}\left(\|\hat{\theta}^n - \theta_0\|_1 \ge \frac{4 \lambda s_0}{\phi^2_0} \right) \le 2 e^{-t^2/2}
\end{equation}
So 
\begin{equation}
    \mathbb{P}\left(\|\Delta\|_\infty \ge \frac{4 \lambda \mu_* s_0 \sqrt{n}}{\phi^2_0}\right) \le 2 e^{-t^2/2}
\end{equation}
The claim follows by selecting $t$ such that $e^{t^2/2}=p^{-c_0}$. Setting $\lambda = \sigma \sqrt{(K^2\log p)/n}$ 
$$
\mathbb{P}\left(\|\Delta\|_\infty \ge \frac{4c s_0}{\phi_0^2} \sqrt{\log p}\right) \le 2 p^{-c_0}
$$
\end{enumerate}
\end{proof}
\noindent
The next theorem establishes that for a natural probabilistic model of the design matrix $\mathbf{X}$, the (constant) bound of the compatibility constant $\phi_0$ and the generalized coherence parameter $\mu_*(\mathbf{X}; M)$ can be bounded with probability converging rapidly to one as $n, p \to \infty$.
\begin{theorem}[Technicality]
\label{thm:thm7}
   Let $\Sigma \in \mathbb{R}^{p\times p}$ be such that $\sigma_{\text{min}}(\Sigma) \ge C_{\text{min}} >0$ and $\sigma_{\text{max}}\le C_{\text{max}}<\infty$ and $\max_{i \in [p]} \Sigma_{ii}\le1$. Assume $\mathbf{X}\Sigma^{-1/2}$ to have independent sub-Gaussian rows, with zero mean and sub-Gaussian norm $\|\Sigma^{-1/2} X_1\|_{\psi_2} = \kappa$, for some constant $\kappa \in (0, \infty)$.
   \begin{enumerate}[label=(\alph*)]
       \item For  $\phi_0$, $s_0$, $K\in \mathbb{R}_{>0}$, let $\mathcal{E}_n = \mathcal{E}_n(\phi_0, s_0, K)$ denote the event that the compatibility condition holds for $\hat{\Sigma} = (\mathbf{X}^\top\mathbf{X}/n)$ for all sets $S \subseteq[p]$, $|S| \le s_0$ with constant $\phi_0>0$ and that $\max_{i\in [p]}\Sigma_{i,i} \le K$. i.e.
       \begin{equation}
           \mathcal{E}_n(\phi_0, s_0, K) \equiv \left\{\mathbf{X}\in \mathbb{R}^{n\times p}: \min_{S: |S| \le s_0} \phi(\hat{\Sigma}, S) \ge \phi_0, \quad \max_{i\in [p]}\hat{\Sigma}_{i,i} \le K,\quad \hat{\Sigma}= (\mathbf{X}^{\top}\mathbf{X}/n)\right\}
       \end{equation}
       Then there exists $x_* \le 2000$ such that the following happens. If $n \ge \nu_0 s_0 \log(p/s_0)$, $\nu_0 \equiv 5 \times 10^4 c_*(C_{\text{max}}/C_{\text{min}})^2 \kappa^4$, $\phi_0 = \sqrt{C_{\text{min}}}/2$ and $K \ge 1 + 20 \kappa^2 \sqrt{(\log p)/n}$, then
       \begin{equation}
           \mathbb{P}(\mathbf{X} \in \mathcal{E}_n(\phi_0, s_0, K)) \ge 1 - 4e^{-c_1n}, \quad c_1 \equiv \frac{1}{4c_* \kappa^4}
       \end{equation}
       \item For $a >0$, let $\mathcal{G}_n= \mathcal{G}_n(a)$ be the event that the convex optimization problem in Algorithm 1 (i.e. (\ref{eq:convex_optimization})) is feasible for $\mu = a \sqrt{(\log p)/n}$. i.e.
       \begin{equation}
           \mathcal{G}_n(a) \equiv \left\{\mathbf{X}\in \mathbb{R}^{n\times p}: \mu_{\text{min}}(\mathbf{X}) < a \sqrt{\frac{\log p}n} \right\}
       \end{equation}
       Then for all $n \ge a^2 C_{\text{min}}\log p /(4e^2C_{\max}\kappa^4)$
       \begin{equation}
               \mathbb{P}(X\in \mathcal{G}_n(a)) \ge 1 - 2p^{-c_2}, \quad c_2 \equiv \frac{a^2 C_{\min}}{24 e^2 \kappa^4 C_{\max}}-2
       \end{equation}
   \end{enumerate}
\end{theorem}
\noindent
Theorem \ref{thm:thm7} shows for a natural probabilistic model of the design matrix $\mathbf{X}\in \mathbb{R}^{n\times p}$, both $\phi_0$ and $\mu_*(\mathbf{X}; M)$ can be bounded with probability that both converge rapidly to 1 as $n,p \to \infty$. Specifically, it guarantees that for sufficiently large $n$, the design matrix remains numerically stable enough to recover the true parameters while maintaining the low coherence required to filter out the shrinkage bias inherent in the LASSO. The proof for this theorem is omitted as it is a more nuanced technicality. For the proof, see the appendix of \citet{javanmard14a}.

Theorems \ref{thm:thm6} and \ref{thm:thm7} establish the necessary probabilistic guarantees for recovering asymptotically unbiased and Gaussian estimates of $\theta_0$. The following theorem builds on these results to characterize the limiting distribution of the debiased estimator $\hat{\theta}^u$ produced by Algorithm 1.

\begin{theorem}[Asymptotic Normality and Error Bounds for the Debiased LASSO Estimator $\hat{\theta}^u$] 
    \label{thm:thm8}
    Consider the linear model (\ref{eq:mod_1dim}) and the debiased LASSO estimator (\ref{eq:lasso_debiased}) defined in Algortihm 1 with $\mu = a \sqrt{(\log p)/n}$. Then setting $Z = M\mathbf{X}^\top W/\sqrt{n}$, we have 
    \begin{equation}
        \sqrt{n}(\hat{\theta}^u - \theta_0) = Z + \Delta, \quad Z|\mathbf{X} \sim N(0, \sigma^2M \hat{\Sigma}M^\top), \quad \Delta = \sqrt{n}(M\hat{\Sigma}- I)(\theta_0 - \hat{\theta}^n)
    \end{equation}
    Further, under the assumptions of Theorem \ref{thm:thm7} for $n \ge \max(\nu_0s_0 \log(p/s_0), \nu_1 \log p)$, $\nu_1= \max(1600 \kappa^4, a^2/(4e^2\kappa^4))$ and $\lambda = \sigma \sqrt{(c^2 \log p)/n}$ 
    \begin{equation}
    \label{eqn:tailbound}
    \mathbb{P}\left\{\|\Delta\|_\infty \ge \left(\frac{16ac \sigma}{C_{\min}} \right)\frac{s_0 \log p}{\sqrt{n}}\right\} \le 4e^{-c_1n} + 4p^{-\min(\tilde{c}_0, c_2)}
    \end{equation}
    Where $\tilde{c}_0 = (c^2/48)-1$, $c_1 = \frac{1}{4c_* \kappa^4}$ and $c_2 = \frac{a^2 C_{\min}}{24 e^2 \kappa^4 C_{\max}}$. 
    \noindent
    
    Finally the tail bound (\ref{eqn:tailbound}) holds for any choice $M$ that is a function of the design matrix and satistfies the feasibility condition of Equation (\ref{eq:convex_optimization}) in Algorithm 1, i.e. $|M\hat{\Sigma}- I|_\infty \le \mu$.
\end{theorem}
Assuming $\sigma$, $C_{\min}$ of order 1, Theorem \ref{thm:thm8} establishes that for random designs, the  maximum size of the bias term $\Delta_i$ over $i \in [p]$ is 
\begin{equation}
\|\Delta\|_{\infty} = O\left(\frac{s_0\log p}{\sqrt{n}}\right)
\end{equation}
On the other hand, the `noise term' $Z_i$ is roughly of order $\sqrt{[M \hat{\Sigma}M^\top]_{ii}}$
\begin{remark}
    Theorem \ref{thm:thm8} only requires that the support size satisfies $s_0 = O(n/\log p)$. If we further assume $s_0 = o(\sqrt{n}/\log p)$ then we have $\|\Delta\|_\infty = o(1)$ with high probability. Hence $\hat{\theta}^u$ is an asymptotically unbiased estimator for $\theta_0$.
\end{remark}
\noindent
We will now present the proof for Theorem \ref{thm:thm8}:
\begin{proof}
    Denote $\|\Delta\|_0 \equiv \left(\frac{16ac \sigma}{C_{\min}} \right)\frac{s_0 \log p}{\sqrt{n}}$. We can then show that 
    $$
\begin{aligned}
    \mathbb{P} \left\{ \|\Delta\|_\infty \ge \left( \frac{16ac \sigma}{C_{\min}} \right) \frac{s_0 \log p}{\sqrt{n}} \right\} &= \mathbb{P} \left\{ \|\Delta\|_\infty \ge \Delta_0 \right\} \\
    (*) \quad \le &\mathbb{P} \Big( \left\{ \|\Delta\|_\infty \ge \Delta_0 \right\} \cap \mathcal{E}_n(\sqrt{C_{\min}}/2, s_0, 3/2) \cap \mathcal{G}_n(a) \Big) \\
    &+ \mathbb{P} \left( \mathcal{E}_n^c(\sqrt{C_{\min}}/2, s_0, 3/2) \right) + \mathbb{P}(\mathcal{G}_n^c(a)) \\
    (**) \quad \le &\mathbb{P} \Big( \left\{ \|\Delta\|_\infty \ge \Delta_0 \right\} \cap \mathcal{E}_n(\sqrt{C_{\min}}/2, s_0, 3/2) \cap \mathcal{G}_n(a) \Big)\\ &+4e^{-c_1n} + 2p^{-c_2}
\end{aligned}
    $$
    Where the first equation $(*)$ comes from the law of total probability and $(**)$ comes from the probabilities of the compliments stated in Theorem \ref{thm:thm7}. We note further that the bound from Theorem \ref{thm:thm7} (b) (i.e. $\mathbb{P}(\mathbf{X} \in \mathcal{E}_n(\phi_0, s_0, K)) \ge 1 - 4e^{-c_1n}, \quad c_1 \equiv \frac{1}{4c_* \kappa^4}$) can be applied for $K= 3/2$ since under present assumptions $K \ge 1 + 20\kappa^2 \sqrt{(\log p)/n} \iff 3/2 \ge 1 + 20\kappa^2 \sqrt{(\log p)/n} \iff  1/2 \ge 20\kappa^2 \sqrt{(\log p)/n}$. 
$$
\begin{aligned}
    \mathbb{P} \Big( \left\{ \|\Delta\|_\infty \ge \Delta_0 \right\} \cap \mathcal{E}_n(\sqrt{C_{\min}}/2, s_0, 3/2) \cap \mathcal{G}_n(a) \Big)     &\le \sup_{\mathbf{X} \in \mathcal{E}_n(\sqrt{C_{\min}}/2, s_0, 3/2) \cap \mathcal{G}_n(a)} \mathbb{P} \left( \|\Delta\|_\infty \ge \Delta_0 \mid \mathbf{X} \right) \\
    &\le 2p^{-\tilde{c}_0}
\end{aligned}
 $$
    Where the last inequality follows from Theorem 6. Thus 
    $$
    \mathbb{P}\left\{\|\Delta\|_\infty \ge \left(\frac{16ac \sigma}{C_{\min}} \right)\frac{s_0 \log p}{\sqrt{n}}\right\} \le 4e^{-c_1n} +2p^{-c_2}+ 2p^{-\tilde{c}_0} \le 4e^{-c_1n} + 4p^{-\min(\tilde{c}_0, c_2)}
    $$
    Where the last inequality follows by bounding the sum by the dominant rate
\end{proof}
\noindent
Theorem \ref{thm:thm8} shows that asymptotic normality and error bounds of the generalized debiased estimator $\hat{\theta}^*$ can be applied to the debiased estimtor proposed by \citet{javanmard14a} in Algorithm 1 $\hat{\theta}^u$ (See Equation (\ref{eq:gen_debiased_theta})). 

\subsection{Statistical Inference}
\label{sec:2.3}
A direct application of Theorem \ref{thm:thm8} is to derive confidence intervals and statistical hypothesis tests for high dimensional models. The next Lemma will show how this is possible. Note, the sparisty assumption is assumed to be $s_0 = o(\sqrt{n}/\log p)$. 
\begin{lemma}[Asymptotic Normality and Valid Inference for the Debiased LASSO]
\label{lemma:lemma13}
    Consider a sequence of design matrices $\mathbf{X} \in \mathbb{R}^{n\times p}$ with dimension $n \to \infty$, $p = p(n)\to \infty$ satisfying the following assumptions for the constants $C_{\min}$,$C_{\max}$, $\kappa \in (0, \infty)$ independent of $n$. For each $n, \Sigma \in \mathbb{R}^{p\times p}$ is such that $\sigma_{\min}(\Sigma) \ge C_{\min} >0$ and $\sigma_{\max}(\Sigma) \le C_{\max} < \infty$ and $\max_{i\in [p]}\Sigma_{ii} \le 1$. Assume $\mathbf{X} \Sigma^{-1/2}$ to have independent sub-Gaussian rows with zero mean and sub-Gaussian norm $\|\Sigma^{-1/2}X_1\|_{\psi_2} \le \kappa$.
    \\
    \\
    Consider the linear model (\ref{eq:mod_1dim}) and let $\hat{\theta}^u$ be the debaised LASSO estimate (\ref{eq:lasso_debiased}) derived in Algorithm 1 with $\mu = a \sqrt{(\log p)/n}$ and $\lambda = \sigma \sqrt{(c^2 \log p)/n}$ with $a$ and $c$ being large enough constants. Finally, let $\hat{\sigma} = \hat{\sigma}(y,\mathbf{X})$ be an estimator of the noise level satisfying for any $\varepsilon>0$
    \begin{equation}
        \lim_{n \to \infty} \sup_{\theta_0 \in \mathbb{R}^p; \|\theta_0\|_0 \le s_0}\mathbb{P}\left(\left|\frac{\hat{\sigma}}{\sigma} - 1\right| \ge \varepsilon \right) = 0
    \end{equation}
    If $s_0 = o(\sqrt{n}/\log p)$ $(s_0 \ge 1)$, then for all $x \in \mathbb{R}$, we have
    \begin{equation}
        \lim_{n \to \infty} \sup_{\theta_0 \in \mathbb{R}^p; \|\theta_0\|_0 \le s_0}\left|\mathbb{P}\left(\frac{\sqrt{n}(\hat{\theta}^u_i - \theta_{0,i})}{\hat{\sigma} [M \hat{\Sigma}M^{\top}]_{i,i}^{1/2}}  \le x \right)- \Phi(x)\right| = 0
    \end{equation}
    Where $\Phi(\cdot)$ denotes the cumulative distribution function (CDF) of a standard normal random variable.
\end{lemma}
While the standard LASSO is an excellent tool for variable selection, its inherent bias makes it unsuitable for traditional statistical inference. Lemma \ref{lemma:lemma13} provides the necessary theoretical bridge, transforming $\hat{\theta}^u$ into a standardized statistic that converges to a standard normal distribution. This in turn shows that it is permissible to perform 'OLS-style' inference in settings where $p \gg n$
\begin{proof}
    Under stated the stated assumptions in Lemma \ref{lemma:lemma13},we can prove 
    \begin{equation}
    \label{eq:lemma13proof1}
        \lim \sup_{n \to \infty}\sup_{\|\theta_0\|_0 \le s_0}\mathbb{P}\left(\frac{\sqrt{n}(\hat{\theta}^u_i - \theta_{0,i})}{\hat{\sigma} [M \hat{\Sigma}M^{\top}]_{i,i}^{1/2}}  \le x \right) \le \Phi(x)
    \end{equation}
    By Theorem \ref{thm:thm8}, we have
    $$
    \frac{\sqrt{n}(\hat{\theta}_i^u - \theta_{0, i})}{\sigma [M \hat{\Sigma} M^\top]^{1/2}_{i,i}} = \frac{e^{\top}_i M \mathbf{X}^\top W}{\sigma [M \hat{\Sigma} M^\top]^{1/2}_{i,i}} + \frac{\Delta_i}{\sigma [M \hat{\Sigma} M^\top]^{1/2}_{i,i}}
    $$
    Let $V = \mathbf{X}M^\top e_i / (\sigma [M \hat{\Sigma} M^\top]_{ii}^{1/2})$ and $\tilde{Z} \equiv V^\top W$. Conditioned on $\mathbf{X}$ (and thus $V$), the independence of the noise $W \sim N(0, \sigma^2 I_n)$ implies:
    $$
    \mathbb{P}( \tilde{Z} \le x) = \mathbb{E}\left\{\mathbb{P}(V^\top W\le x|V) \right\} = \mathbb{E}\left\{\Phi(x) | V \right\} = \Phi(x)
    $$
    In order to prove Equation (\ref{eq:lemma13proof1}), fix $\varepsilon >0$ and write
    $$
    \begin{aligned}
    \mathbb{P}\left(\frac{\sqrt{n}(\hat{\theta}^u_i - \theta_{0,i})}{\hat{\sigma} [M \hat{\Sigma}M^{\top}]_{i,i}^{1/2}}  \le x \right) &= \mathbb{P}\left(\frac{\sigma}{\hat{\sigma}}\tilde{Z} +  \frac{\Delta_i}{\hat{\sigma}[M \hat{\Sigma}M^\top]^{1/2}_{i,i}}  \le x \right)\\
    &\le \mathbb{P}\left(\frac{\sigma}{\hat{\sigma}}\tilde{Z} \le x + \varepsilon \right) + \mathbb{P}\left( \frac{\Delta_i}{\hat{\sigma}[M \hat{\Sigma}M^\top]^{1/2}_{i,i}}  \ge \varepsilon\right)\\
    &\le \mathbb{P}\left(\tilde{Z} \le x + 2\varepsilon + \varepsilon|x| \right) + \mathbb{P}\left( \frac{\Delta_i}{\hat{\sigma}[M \hat{\Sigma}M^\top]^{1/2}_{i,i}}  \ge \varepsilon\right) + \mathbb{P}\left(\left|\frac{\hat{\sigma}}{\sigma} - 1\right| \ge \varepsilon \right)
    \end{aligned}
    $$
    Taking the limit and using assuming that $\lim_{n \to \infty} \sup_{\theta_0 \in \mathbb{R}^p; \|\theta_0\|_0 \le s_0}\mathbb{P}\left(\left|\frac{\hat{\sigma}}{\sigma} - 1\right| \ge \varepsilon \right) = 0$ we obtain:
    $$
          \lim_{n \to \infty} \sup_{\theta_0 \in \mathbb{R}^p; \|\theta_0\|_0 \le s_0}\mathbb{P}\left(\frac{\sqrt{n}(\hat{\theta}^u_i - \theta_{0,i})}{\hat{\sigma} [M \hat{\Sigma}M^{\top}]_{i,i}^{1/2}}  \le x \right) \le \Phi(x + 2\varepsilon + \varepsilon |x| ) + \lim_{n \to \infty} \sup_{\theta_0 \in \mathbb{R}^p; \|\theta_0\|_0  \le s_0} \mathbb{P}\left( \frac{\Delta_i}{\hat{\sigma}[M \hat{\Sigma}M^\top]^{1/2}_{i,i}}  \ge \varepsilon\right)
    $$
Since $\varepsilon > 0$ is arbitrary, showing that the limit vanishes requires bounding the absolute difference from both sides. By the properties of probability measures for any $\varepsilon > 0$, we have:
\begin{equation}
    \left| \mathbb{P}\left( \frac{\sqrt{n}(\hat{\theta}^u_i - \theta_{0,i})}{\hat{\sigma} [M \hat{\Sigma}M^{\top}]_{i,i}^{1/2}} \le x \right) - \Phi(x) \right| \le \mathbb{P}\left( \frac{|\Delta_i|}{\hat{\sigma}[M \hat{\Sigma}M^\top]^{1/2}_{i,i}} \ge \varepsilon \right) + \delta_\varepsilon(x)
\end{equation}
where $\delta_\varepsilon(x)$ represents the modulus of continuity of $\Phi(\cdot)$, which vanishes as $\varepsilon \to 0$. By Theorem \ref{thm:thm8}, the first term on the right-hand side converges to zero as $n \to \infty$. Since this holds for any $\varepsilon$, the supremum of the absolute difference over the sparsity class $s_0$ must likewise converge to zero.
\end{proof}
Lemma \ref{lemma:lemma13} is the bridge which facilitates standard inferential techniques such as confidence intervals and hypothesis testing for $\hat{\theta}^u$. Furthermore, this result is generalizable to simultaneous inference and remains robust under non-Gaussian noise distributions. These results are presented in subsections \ref{sec:2.3.1},\ref{sec:2.3.2},\ref{sec:2.3.3} and \ref{sec:2.3.4} without proof for brevity.
\subsubsection{Confidence Intervals}
\label{sec:2.3.1}
From Lemma \ref{lemma:lemma13} it is quite straightforward to construct asymptotically valid confidence intervals. For $i \in [p]$ and significance level $\alpha \in (0,1)$, the confidence interval of the point estimate of the debaised LASSO estimator $\hat{\theta}^u$ is defined as:
\begin{equation}
\label{eq:ci_formula}
\begin{aligned}
    J_i(\alpha) &\equiv \left[ \hat{\theta}_i^u - \delta(\alpha, n), \hat{\theta}_i^u + \delta(\alpha, n) \right], \\
    \delta(\alpha, n) &\equiv \Phi^{-1}(1 - \alpha/2) \frac{\hat{\sigma}}{\sqrt{n}} [M \hat{\Sigma} M^\top]_{i,i}^{1/2}
\end{aligned}
\end{equation}
It can also be shown that the confidence interval is asymptotically valid, namely 
\begin{equation}
    \lim_{n \to \infty} \mathbb{P}\left(\theta_{0,i} \in J_i(\alpha) \right) = 1-\alpha
\end{equation}
\subsubsection{Hypothesis Testing}
\label{sec:2.3.2}
An important advantage of sparse linear regression methods such as the LASSO is their ability to provide a parsimonious explanation of the data, relying on only a small number of covariates. This can be simply understood as selecting the number of non-zero covariates $\hat{\theta}^n_i \neq 0$ for all $i \subseteq [p]$. However, the covariates estimated by the LASSO $\hat{\theta}^n_i$ do not provide a measure of the statistical significance for the finding that the coefficient is non-zero. In particular, we are interested in testing an individual null hypotheses $H_{0,i}: \theta_{0,i} = 0$ verses the alternative $H_{A,i}: \theta_{\theta,i} \neq 0$. Using the debiased estimator $\hat{\theta}^u$ derived in Algorithm 1, we construct the p-value $P_i$ for the test $H_{0,i}$ to be: 
\begin{equation}
    P_{i} = 2 \left(1 - \Phi\left(\frac{\sqrt{n}|\hat{\theta}^{u}_i |}{\hat{\sigma}[M \hat{\Sigma M^\top}]^{1/2}_{i,i}}\right) \right)
\end{equation}
The decision rule is based on the p-value $P_i$:
\begin{equation}
    \widehat{T}_{i, \mathbf{X}}(y) = \begin{cases}
        1 & \text{if } P_i \le \alpha \quad (\text{reject } H_{0,i}), \\
        0 & \text{otherwise} \quad (\text{accept } H_{0,i})
    \end{cases}
\end{equation}
Where $\alpha$ is the fixed target Type I error probability. We measure the quality of the test $ \widehat{T}_{i, \mathbf{X}}(y)$ in terms of its significance level $\alpha_i$ and statistical power $1- \beta_i$ where $\beta_i$ is the probability of type II error. It can be shown that this test achieves power $1 - \beta =\alpha$. From a formal minimax perspective, given a family of tests $T_{i, \mathbf{X}}:\mathbb{R}^n \to \{ 0, 1\}$ indexed by $i \in [p]$, $\mathbf{X} \in \mathbb{R}^{n\times p}$, we defined for $\gamma >0$ a lower bound on the non-zero entries of $\theta_{0,i}$
\begin{equation}
    \begin{aligned}
        \alpha_{i, n}(T) &\equiv \sup \left\{\mathbb{P}_{\theta_0}(T_{i,\mathbf{X}}(y)= 1):\theta_0 \in \mathbb{R}^p, \|\theta_0\|_0 \le s_0(n), \theta_{0,i} = 0 \right\}\\
        \beta_{i, n}(T; \gamma) & \equiv\sup \left\{\mathbb{P}_{\theta_0}(T_{i,\mathbf{X}}(y)= 0):\theta_0 \in \mathbb{R}^p, \|\theta_0\|_0 \le s_0(n), \theta_{0,i} \ge \gamma \right\}
    \end{aligned}
\end{equation}
\subsubsection{Generalizing to Simultaneous Confidence Intervals}
\label{sec:2.3.3}
In many situations, it is necessary to perform statistical inference on more than one of the parameters simultaneously. As an example, we might be interested in performing inference about $\theta_{0,R}\equiv (\theta_{0,i})_{i \in R}$ for some set $R \subseteq [p]$. The simplest generalization of the method method is when the number of indices $|R|$ stay finite as $n, p \to \infty$. The following Lemma is a generalization of Lemma \ref{lemma:lemma13}.
\begin{lemma}
\label{lemma:lemma19}
Under the assumptions of Lemma \ref{lemma:lemma13}, define 
\begin{equation}
Q^{(n)} \equiv \frac{\hat{\sigma}^2}{n}[M \hat{\Sigma}M^\top]
\end{equation}
Let $R=R(n)$ be a sequence of sets $R(n) \subseteq [p],$ with $|R(n)| = k$ fixed as $n, p \to \infty$ and further assume $s_0 = o(\sqrt{n}/\log p)$ with $s_0\ge1$. Then for all $x = (x_1, \dots, x_k) \in \mathbb{R}^k$, we have
\begin{equation}
    \lim_{n \to \infty}\sup_{\theta_0 \in \mathbb{R}^p; \|\theta_0\|_0 \le s_0}\left|\mathbb{P}\left\{\left(Q^{(n)}_{R,R}\right)^{-1/2} (\hat{\theta}^u_R - \theta_{0, R}) \le x \right\}  -\mathbf{\Phi}_k(x)\right| =0
\end{equation}
where $(a_1, \dots, a_k) \le (b_1, \dots, b_k)$ indicates $a_1 \le b_1, \dots, a_k \le b_k$, and $\mathbf{\Phi}_k(x) = \Phi(x_1) \times \dots \times\Phi(x_k)$
\end{lemma}
\noindent
Mirroring the univariate inferential capabilities provided by Lemma \ref{lemma:lemma13}, Lemma \ref{lemma:lemma19} provides the necessary theoretical results to move from marginal to simultaneous confidence regions. We can then define the valid confidence region $J_R(\alpha) \subseteq \mathbb{R}^k$ as 
\begin{equation}
    J_{R}(\alpha) \equiv \hat{\theta}_R^u + (Q^{(n)}_{R,R})^{1/2}\mathcal{C}_{k,\alpha}
\end{equation}
Such that 
\begin{equation}
    \lim_{n \to \infty} \mathbb{P}\left(\theta_{0, R} \in J_R(\alpha)\right)= 1-\alpha
\end{equation}
Where $\mathcal{C}_{k, \alpha}\subseteq \mathbb{R}^k$ is any Borel set such that $\int_{\mathcal{C}_{k, \alpha}} \phi(x) dx \ge 1-\alpha$ and $\phi(x) = \frac{1}{(2\pi)^{k/2}}\exp \left(- \frac{\|x\|^2}{2}\right)$. We can further construct a testing procedure that controls the family wise error rate (FWER) based on Bonferroni.
\begin{equation}
\widehat{T}^F_{i, \mathbf{X}}(y) = \begin{cases}
1 \quad \text{if }P_i \le \alpha/p \quad (\text{reject } H_{0,i})\\
0 \quad \text{otherwise} \quad (\text{accept } H_{0,i})
\end{cases}
\end{equation}
It then follows from the Bonferroni inequality that the asymptotic family-wise error rate is bounded above by $\alpha$:
\begin{equation}
\limsup_{n\to\infty} \text{FWER}(\widehat{T}^F,n) \le \alpha    
\end{equation}
\subsubsection{Non-Gaussian Noise}
\label{sec:2.3.4}
The inferential guarantees for the simultaneous confidence region $J_R(\alpha)$, the individual hypothesis tests, and the FWER-controlling procedure $\widehat{T}^F$ naturally extend to non-Gaussian noise, provided the independent errors $W_i$ have zero mean, variance $\sigma^2$, and a bounded moment $\mathbb{E}[|W_i|^{2+a}] \le C\sigma^{2+a}$ for some $a > 0$. To ensure the asymptotic normality required for valid confidence intervals and $p$-values, the scaled summands of the error, denoted as $\xi_j=m_i^\top X_j W_j/(\sigma[m_i^\top \hat{\Sigma}m_i]^{1/2})$, must satisfy the Lindeberg condition. Namely, for every $\varepsilon > 0$, almost surely:
\begin{equation}
\lim_{n \to \infty} \frac{1}{n} \sum_{j=1}^n \mathbb{E}\left(\xi_j^2 \mathbb{I}_{\{|\xi_j| > \varepsilon\sqrt{n}\}} \mid X\right) = 0    
\end{equation}
Then $\sum_{j=1}^n \xi_k/\sqrt{n} \overset{d}\to N(0,1)$. This condition is satisfied by introducing an $\ell_\infty$ constraint to the debiasing optimization problem: $\|Xm_i\|_\infty \le n^\beta$ for a fixed $\beta \in (1/4, 1/2)$. By bounding the maximum influence of any single observation, the Lindeberg condition holds, preserving the asymptotic Gaussianity that guarantees the validity of the intervals, the tests, and the FWER upper bound $\limsup_{n\to\infty} \text{FWER}(\widehat{T}^F,n) \le \alpha$. 

To construct the debiasing matrix $M = (m_1, \dots, m_p)^T$ under the non-Gaussian regime, each row vector $m_i \in \mathbb{R}^p$ is obtained by solving the following convex optimization problem:
\begin{equation}
    \begin{aligned}
& \underset{m \in \mathbb{R}^p}{\text{minimize}} & & m^T \hat{\Sigma} m \\
& \text{subject to} & & \|\hat{\Sigma}m - e_i\|_\infty \le \gamma, \\
& & & \|Xm\|_\infty \le n^\beta \quad \text{ for arbitrary fixed } 1/4 <\beta<1/2
\end{aligned}
\end{equation}
Consequently, under the assumptions of Theorem \ref{thm:thm8}, the resulting p-values remain asymptotically valid in non-Gaussian settings; a formal statement and proof are omitted for brevity.
\section{Empirical Analysis}
\label{sec:empirical_analysis}
To evaluate the theoretical results of \citet{javanmard14a}, we replicate six of the original simulation scenarios (selected due to computational constraints) and the real-data application on the riboflavin (vitamin $B_2$) dataset. Building on this baseline, we compare their method with the Lasso projection estimator \citep{vandegeer2014asymptotically, zhang2014confidence, buhlmann2015misspecified} (also known as the LASSO projection estimator), which was cited but not implemented in the original study. As the original work only benchmarked against multisample splitting and a ridge projection estimator, this additional comparison provides a more comprehensive assessment of its practical performance.

All computations were performed on a first-generation Intel NUC laptop equipped with a 12th Gen Intel Core i7-1260P processor. The replication study was conducted using \texttt{R} version 4.2.2, leveraging the \texttt{parallel}, \texttt{doParallel}, and \texttt{foreach} packages to enable parallel processing \citep{rlanguage, doParalllelpkg, foreachpkg}. The implementation of the debiased LASSO proposed by \citet{javanmard14a} was sourced from the authors' repository (\href{https://web.stanford.edu/~montanar/sslasso/}{\texttt{https://web.stanford.edu/~montanar/sslasso/}}). Furthermore, routines for multisample splitting, the ridge projection estimator, and the LASSO projection estimator, alongside the riboflavin dataset, were accessed via the \texttt{hdi} \texttt{R} package \citep{hdi_pkg2015}.
 
\subsection{Simulation Study}
\label{sec:3.1}
We consider the linear model (\ref{eq:mod_mat_form}) where rows of the design matrix $\mathbf{X}$ are fixed i.i.d. realizations from the  $N(0, \Sigma)$, where $\Sigma \in \mathbb{R}^{p\times p}$ is a circulant symmetric matrix with entries $\Sigma_{jk}$ given as for $j \le k$
\begin{equation}
\label{eq:circulant_mat}
    \Sigma_{jk} = \begin{cases}
    1 &\text{if } k= j\\
    0.1 & \text{if } k \in\{j+1, \dots, j+5 \}\\
   & \text{or } k \in \{j+p-5, \dots, j+p-1\}\\
    0 & \text{for all other } j\le k.
    \end{cases}
\end{equation}

Figure \ref{fig:circulant_mat} illustrates an example of this matrix structure for $p=600$. To generate the true regression coefficients, we select a uniformly random support set $S \subseteq [p]$ of size $|S|=s_0$. The non-zero coefficients are set to $\theta_{0,i}=b$ for $i \in S$, and $\theta_{0,i}=0$ otherwise. Measurement errors are drawn as $W_i \sim N(0,1)$ for $i \in [n]$. We evaluate several configurations of $(n, p, s_0, b)$, and for each configuration, the results are averaged over 20 independent realizations of the model with fixed design and regression coefficients. In contrast to \citet{javanmard14a}, who explicitly tune a regularization parameter and model optimization constraint $\mu$, we rely on the default configurations for all methods to better reflect how these implementations are typically used in practice.

For each configuration, 20 independent realizations of measurement noise and for each parameter $\theta_{0,i}$ the average length of the corresponding confidence interval is computed. Denoted by $\text{Avglength}(J_i(\alpha))$ where $J_i(\alpha)$ is the estimated confidence interval for $\hat{\theta}_i^u$ defined in Equation (\ref{eq:ci_formula}) and the average is taken over all realizations. The following metric considered is
$$
\ell \equiv p^{-1} \sum_{i \in [p]} \text{Avglength}(J_i(\alpha))
$$
Along with the average length of intervals for the active and inactive parameters as:
\begin{equation}
    \ell_S \equiv s_0^{-1} \sum_{i \in S}\text{Avglength}(J_i(\alpha)), \quad \ell_S^c \equiv (p-s_0)^{-1}\sum_{i \in S^c}\text{Avglength}(J_i(\alpha))
\end{equation}
Similarly, the average coverage for the individual parameters is considered. i.e.
\begin{equation}
\widehat{\text{Cov}} = p^{-1}\sum_{i \in [p]}\widehat{\mathbb{P}}[ \theta_{0,i}\in J_i(\alpha)],    
\end{equation}
\begin{equation}
\widehat{\text{Cov}}_S = s_0^{-1}\sum_{i \in S}\widehat{\mathbb{P}}[ \theta_{0,i}\in J_i(\alpha)],
\end{equation}
\begin{equation}
\widehat{\text{Cov}}_{S^c} = s_0^{-1}\sum_{i \in S^c}\widehat{\mathbb{P}}[ \theta_{0,i}\in J_i(\alpha)]
\end{equation}
Where probabilities ($\widehat{\mathbb{P}}$) are calculated from 20 realizations per configuration. Tables \ref{tab:tab1}--\ref{tab:tab4} compare interval coverage, false positive (FP) rates, and statistical power (TP) between this replication and \citet{javanmard14a}. Tables \ref{tab:tab2} and \ref{tab:tab4} show the relevant results listed in \citet{javanmard14a}.

As detailed in Tables \ref{tab:tab1} and \ref{tab:tab2}, the replication produces tighter confidence intervals (e.g., $\ell \approx 0.15$ versus $\ell \approx 0.18$) and achieves a more precise alignment with the 95\% nominal coverage compared to the original results, which tended toward over-coverage. When evaluating error rates (Tables \ref{tab:tab3} and \ref{tab:tab4}), both studies successfully control FP rates near the nominal $\alpha = 0.05$ level. However, the replication achieves substantially higher statistical power (TP) across multiple methods in moderately sparse settings, most notably maintaining full power for multisample-splitting where the original study observed a significant drop.

Crucially, Table \ref{tab:tab3} highlights the performance of the LASSO projection estimator, which was not considered in the original study despite being cited. In direct comparison to the replicated Javanmard and Montanari method, the LASSO projection demonstrates superior sensitivity in low-signal regimes. For instance, under the $(1000, 600, 30, 0.1)$ configuration, the LASSO projection yields a TP of 0.8217 compared to the Javanmard and Montanari method's 0.7017, while still rigorously controlling the FP rate (0.0479 versus 0.0363). These findings indicate that the LASSO projecton estimator provides a more powerful framework for high-dimensional inference than the original Javanmard and Montanari procedure, improving signal detection without inflating Type I error.
\begin{figure}[h]
    \centering
    \includegraphics[width=\textwidth]{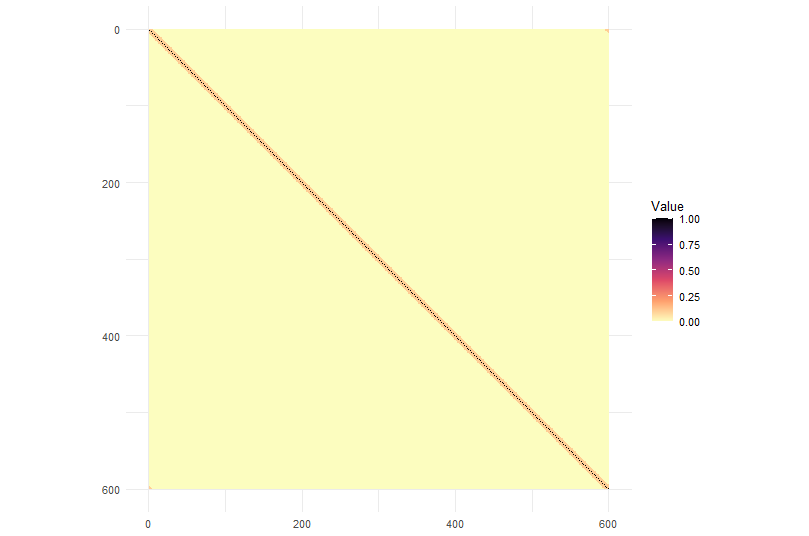}
    \caption{A visualization of the symmetric circulant matrix $\Sigma$ specified mathematically in Equation (\ref{eq:circulant_mat})}
    \label{fig:circulant_mat}
\end{figure}

\begin{table}[ht]
\centering
\begin{tabular}{|l|c|c|c|c|c|c|}
  \hline
\diagbox{Configuration}{Measure} & $\ell$ & $\ell_S$ & $\ell_{S^c}$ & $\widehat{Cov}$ & $\widehat{Cov}_S$ & $\widehat{Cov}_{S^c}$ \\ 
  \hline
(1000,600,10,0.5) & 0.1485 & 0.1485 & 0.1485 & 0.9592 & 0.9450 & 0.9594 \\ 
  (1000,600,10,0.25) & 0.1473 & 0.1471 & 0.1473 & 0.9583 & 0.9600 & 0.9583 \\ 
  (1000,600,10,0.1) & 0.1518 & 0.1521 & 0.1518 & 0.9602 & 0.9850 & 0.9597 \\ 
  (1000,600,30,0.5) & 0.1528 & 0.1525 & 0.1528 & 0.9657 & 0.9617 & 0.9659 \\ 
  (1000,600,30,0.25) & 0.1527 & 0.1528 & 0.1527 & 0.9680 & 0.9750 & 0.9676 \\ 
  (1000,600,30,0.1) & 0.1590 & 0.1588 & 0.1590 & 0.9711 & 0.9600 & 0.9717 \\ 
   \hline
\end{tabular}
\caption{Replicated results for synthetic data. Results correspond to 95\% confidence intervals.} 
\label{tab:tab1}
\end{table}

\begin{table}[ht]
\centering
\begin{tabular}{|l|c|c|c|c|c|c|}
\hline
\diagbox{Configuration}{Measure} & $\ell$ & $\ell_S$ & $\ell_{S^c}$ & $\widehat{\text{Cov}}$ & $\widehat{\text{Cov}}_S$ & $\widehat{\text{Cov}}_{S^c}$ \\ \hline
(1000, 600, 10, 0.5) & 0.1870 & 0.1834 & 0.1870 & 0.9766 & 0.9600 & 0.9767 \\
(1000, 600, 10, 0.25) & 0.1757 & 0.1780 & 0.1757 & 0.9810 & 0.9000 & 0.9818 \\
(1000, 600, 10, 0.1) & 0.1809 & 0.1823 & 0.1809 & 0.9760 & 1 & 0.9757 \\
(1000, 600, 30, 0.5) & 0.2107 & 0.2108 & 0.2107 & 0.9780 & 0.9866 & 0.9777 \\
(1000, 600, 30, 0.25) & 0.1956 & 0.1961 & 0.1956 & 0.9660 & 0.9660 & 0.9659 \\
(1000, 600, 30, 0.1) & 0.2023 & 0.2043 & 0.2023 & 0.9720 & 0.9333 & 0.9732 \\
\hline
\end{tabular}
\caption{Simulation results for the synthetic data described in \citep{javanmard14a}. The results corresponds to 95\% confidence intervals.}
\label{tab:tab2}
\end{table}
\begin{table}[ht]
\centering
\begin{tabular}{|c|c|c|c|c|c|c|c|c|}
  \hline & \multicolumn{2}{c|}{ \makecell{Javanmard \\ \& Montanari}} & \multicolumn{2}{c|}{ \makecell{Multisample- \\ splitting}} & \multicolumn{2}{c|}{ \makecell{Ridge projection  \\ estimator}} & \multicolumn{2}{c|}{ \makecell{LASSO projection \\  estimator}} \\ \hline Configuration & FP & TP & FP & TP & FP & TP & FP & TP \\ 
  \hline
  \hline
(1000,600,10,0.5) & 0.0514 & 1 & 0 & 1 & 0.0420 & 1 & 0.0488 & 1 \\ 
  (1000,600,10,0.25) & 0.0547 & 1 & 0 & 1 & 0.0385 & 1 & 0.0490 & 1 \\ 
  (1000,600,10,0.1) & 0.0495 & 0.7700 & 0 & 0 & 0.0454 & 0.4750 & 0.0508 & 0.8550 \\ 
  (1000,600,30,0.5) & 0.0531 & 1 & 0 & 1 & 0.0352 & 1 & 0.0418 & 1 \\ 
  (1000,600,30,0.25) & 0.0518 & 1 & 0 & 0.9967 & 0.0413 & 0.9967 & 0.0433 & 1 \\ 
  (1000,600,30,0.1) & 0.0363 & 0.7017 & 0 & 0.0200 & 0.0384 & 0.4383 & 0.0479 & 0.8217 \\ 
   \hline
\end{tabular}
\caption{Replicated simulation results for false postive rates (FP) and and true positive rates computed at significance level $\alpha = 0.05$.} 
\label{tab:tab3}
\end{table}

\begin{table}[ht]
\centering
\begin{tabular}{|c|c|c|c|c|c|c|}
 \hline & \multicolumn{2}{c|}{ \makecell{Javanmard \\ \& Montanari}} & \multicolumn{2}{c|}{ \makecell{Multisample- \\ splitting}} & \multicolumn{2}{c|}{ \makecell{Ridge projection  \\ estimator}} \\ \hline 
Configuration & FP & TP & FP & TP & FP & TP \\ \hline
(1000, 600, 10, 0.5) & 0.0452 & 1 & 0 & 1 & 0.0284 & 0.8531 \\
(1000, 600, 10, 0.25) & 0.0393 & 1 & 0 & 0.4 & 0.02691 & 0.7506 \\
(1000, 600, 10, 0.1) & 0.0383 & 0.8 & 0 & 0 & 0.2638 & 0.6523 \\
(1000, 600, 30, 0.5) & 0.0433 & 1 & 0 & 1 & 0.0263 & 0.8700 \\
(1000, 600, 30, 0.25) & 0.0525 & 1 & 0 & 0.4 & 0.2844 & 0.8403 \\
(1000, 600, 30, 0.1) & 0.0402 & 0.7330 & 0 & 0 & 0.2238 & 0.6180 \\
\hline
\end{tabular}
\caption{Simulation results for the synthetic data described in \citep{javanmard14a}. The false positive rates (FP) and the true positive rates (TP) are computed at significance level $\alpha = 0.05$.}
\label{tab:tab4}
\end{table}

\subsection{Data Analysis}
\label{sec:3.2}
To demonstrate our approach on empirical data, we analyze the riboflavin production dataset previously examined by \citet{javanmard14a}. Originally introduced by \citet{buhlmann2014high}, the dataset comprises $n=71$ samples and $p=4,088$ covariates representing individual gene expression levels. The continuous response variable measures the log-transformed riboflavin production rate, which is modelled linearly against the log-transformed gene expression covariates.

Figure \ref{fig:data_analysis} shows that Multisample-splitting proved overly conservative with estimated confidence intervals which covered the whole range of the covariate distribution, whereas the debiased projection estimators and the \citet{javanmard14a} approach successfully identified globally significant genes such as \texttt{ARGF\_at} and \texttt{LYSC\_at}. The Javanmard and Montanari method exhibited the highest precision , consistently yielding the narrowest confidence intervals and isolating signal where the Ridge and Lasso projection estimators were less sensitive. These results highlight the substantial precision gains achievable through the optimized de-biasing procedure proposed by \citet{javanmard14a}.

\begin{figure}[h]
    \centering
    \includegraphics[width=\textwidth]{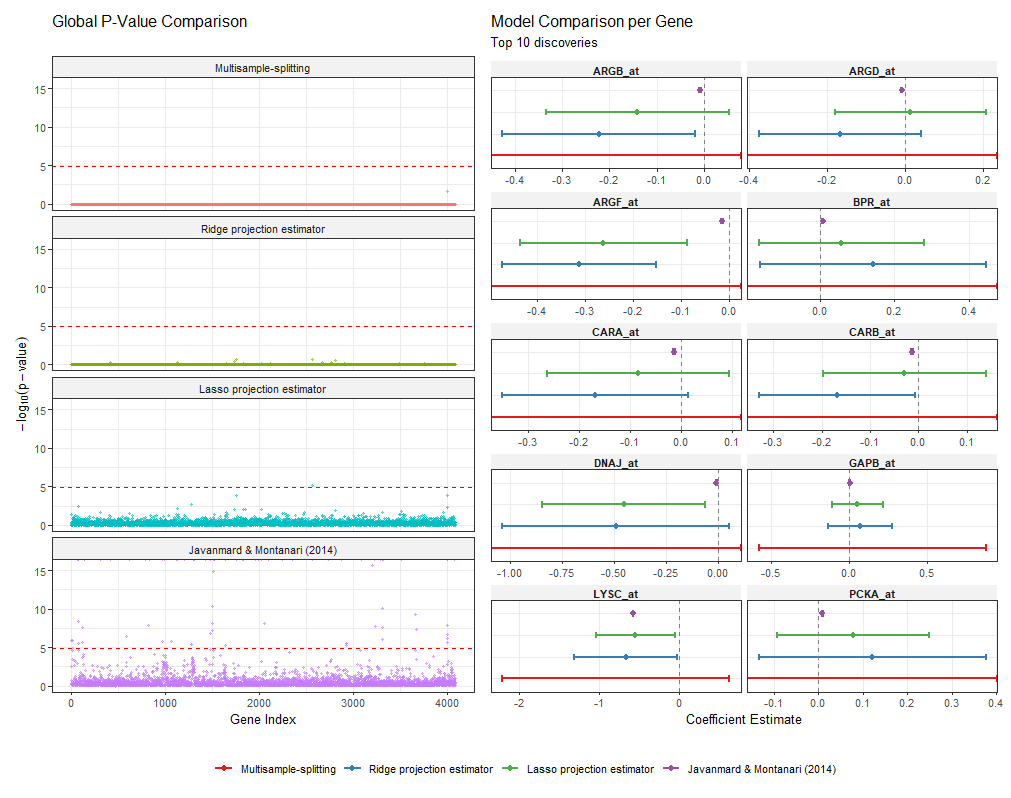}
    \caption{Comparative high-dimensional inference on the riboflavin dataset ($n=71, p=4,088$). The Manhattan plots (left) display the global p-value distribution relative to a Bonferroni-corrected threshold (red line), while faceted forest plots (right) provide 95\% confidence intervals for the top 10 genes.}
    \label{fig:data_analysis}
\end{figure}

\section{Discussion \& Conclusion}
\label{sec:discussion}
The transition of high-dimensional regression from pure variable selection to formal statistical inference marks a critical methodological advancement. Because the non-linear $\ell_1$ penalty of the standard LASSO introduces inherent shrinkage bias, classical uncertainty quantification via confidence intervals and p-values was precluded. The framework proposed by \citet{javanmard14a} successfully bridges this gap by constructing an optimal de-biasing matrix that corrects the LASSO estimate while explicitly minimizing asymptotic variance. The resulting theoretical guarantees for the estimator $\hat{\theta}^u$ related to its asymptotic normality, resilience to non-Gaussian noise, and adaptability for simultaneous testing provide a rigorous mathematical foundation for classical statistical inference in high-dimensional setting when $p \gg n$, under suitable sparsity assumptions.

Empirical evaluations across replications of the simulation study and the riboflavin dataset reveal a nuanced narrative regarding the practical performance of these estimators, highlighting a critical divergence between idealized settings and severely sample-starved, real-world applications. In the simulation study, multisample-splitting \citep{Meinshausen_Buhlman_2009} maintained strong statistical power, and the LASSO projection estimator \citep{zhang2014confidence} outperformed the \citet{javanmard14a} method by demonstrating superior sensitivity in low-signal regimes. However, this dynamic reversed in a real high-dimensional genomic data setting. Furthermore, the Javanmard and Montanari estimator achieved the highest precision and successfully isolated empirical signals where the LASSO projection lacked sensitivity. This discrepancy suggests that while standard projection estimators excel under well-behaved, circulant covariance structures, the explicit convex optimization of the de-biasing matrix $M$ employed by \citet{javanmard14a} adapts far more robustly to the complex correlation networks inherent in real biological data. While modern projection methods offer robust power in idealized scenarios, optimized de-biasing matrices produced by \citep{javanmard14a} prove to be dominant for maximizing signal detection and enabling inference in noisy, high-dimensional settings.

To provide a more equitable evaluation of these estimators, it is recommended that future simulation studies employ design matrices that intentionally weaken or violate the restricted eigenvalue and compatibility conditions. The circulant covariance structure used in the simulation study features minimal off-diagonal correlation (as shown in Figure \ref{fig:circulant_mat}), creating an idealized setting where the LASSO projection estimator excels. Furthermore, these settings fail to capture the complex correlation structures inherent in real-world biological data. By benchmarking against highly correlated data generating processes, future research can more accurately demonstrate the practical necessity of the debaised estimator proposed by \citet{javanmard14a}.
\clearpage
\bibliographystyle{apalike}
\bibliography{references}

\newpage


\label{supp:2} 
\end{document}